\documentclass[10pt]{article}
\usepackage[letterpaper,top=0.68in,bottom=0.72in,left=0.68in,right=0.68in,columnsep=0.24in]{geometry}
\usepackage{graphicx}
\graphicspath{{figures/}}
\usepackage{amsmath,amssymb}
\usepackage{bm}
\usepackage[utf8]{inputenc}
\usepackage[T1]{fontenc}
\usepackage{mathptmx}
\usepackage{microtype}
\usepackage[numbers,sort&compress]{natbib}
\usepackage{amsthm}
\usepackage{url}
\usepackage[ruled,vlined,linesnumbered,algo2e]{algorithm2e}
\usepackage{balance}
\newtheorem{proposition}{Proposition}
\newtheorem{definition}{Definition}
\newcommand{\ones}{\mathbf 1}
\newcommand{\D}{\Delta}
\newcommand{\R}{\mathbb R}
\newcommand{\onlinecite}[1]{\citenum{#1}}

\begin{document}
\title{Structured Stochastic Representations of Integrated Dynamic Strategies}
\author{Fredy Vides\\[0.25em]
\small Department of Applied Mathematics, School of Mathematics and Computer Science,\\[-0.1em]
\small Universidad Nacional Aut\'onoma de Honduras (UNAH)}
\date{9 September 2026}
\twocolumn[
\begin{@twocolumnfalse}
\maketitle
\vspace{-1.6em}
\begin{abstract}
Dynamic allocation decisions couple present resource use to evolving internal conditions, delayed returns, and future costs. We represent this interaction by four probability localizations linked through regime-indexed, graph-constrained column-stochastic operators. Pre-action state or context selects a locally affine model, while action-dependent changes update subsequent regimes, yielding a causal switched representation of nonlinear evolution. We characterize operator identifiability relative to the graph, the stochastic constraints, and the sampled embedding, separating coefficient recovery from predictive equivalence on the decision domain. Decision making is then formulated through implementable return--cost acceptability regions. Finite-horizon error propagation supplies conservative classification margins, and simultaneous intervals distinguish model-relative near-optimality from certified $\epsilon$-optimality over a declared finite policy class. Regime-indexed stochastic feedback is admitted when it satisfies the same certification test. Reproducible synthetic laboratories for personal preparation, supplier participation, and customer retention illustrate exact, operator-supplied, and noisy feedback cases. Multinomial experiments show improving recovery of the feedback function and fewer unresolved decisions with increasing sample size, while unrestricted off-policy recovery remains limited. The contribution is a structure-preserving representation--identification--decision workflow, not a domain-specific physiological or commercial calibration.
\end{abstract}
\vspace{0.8em}
\end{@twocolumnfalse}
]
\section{Introduction}
Allocation is intrinsically dynamic. A present choice changes not only an immediate output but also the internal configuration from which later choices, returns, and costs arise. The consequences depend on history, operating conditions, resource availability, and the order of actions. This structure appears in preparation and recovery, inventory replenishment, customer retention, and other human--cyber--physical or organizational processes. A useful representation should therefore preserve admissible transitions and uncertainty while remaining explicit enough to support an implementable decision.

Probability vectors provide one such interface between observed collective frequencies and dynamical models. Recent stochastic reservoir computers likewise use probabilities of reservoir outcomes as trainable readouts and describe their evolution through controlled Markov dynamics; universality results establish the approximation capacity of particular stochastic echo-state-network classes~\cite{ehlers2025stochastic}. The present objective is different. We do not seek universality or a physical reservoir architecture. We seek graph-structured stochastic operators whose predictions can be connected to return--cost decisions and whose nonidentifiable directions are reported rather than hidden. The construction extends the probability-localization viewpoint of Vides~\cite{vides2026localization} toward allocation and feedback, using the structured stochastic identification architecture of Banegas and Vides~\cite{banegas2025ssrc}.

Nonlinear and nonstationary behavior is represented by regime-indexed local models. Switched and hybrid-system formulations provide established tools for describing mode-dependent dynamics~\cite{hamilton1989regime,liberzon2003switching,paoletti2007hybrid}. Here the regime is selected causally from pre-action information, or updated for the following period when the current action affects it. Each active probability map remains affine, but state- or context-dependent selection produces nonlinear global evolution. This moderate use of switching avoids presenting one fitted matrix as a universal description of all operating conditions.

Identification and decision are deliberately separated. Classical system identification emphasizes excitation, model structure, and the special difficulties created by closed-loop data~\cite{ljung1999system,forssell1999closedloop}. Those difficulties are central here because a behavioral allocation law can correlate actions with states and leave feasible operator directions observationally equivalent. Exact coefficient recovery is therefore stronger than predictive adequacy on a declared policy domain. On the decision side, stochastic model predictive control and chance-constrained methods optimize expected performance while managing probabilistic violations~\cite{mesbah2016smpc}. Our aim is narrower and complementary: classify policies against explicit achievement and cost thresholds, retain unresolved cases when error intervals cross a boundary, and optimize only within the certified region. This implements a dynamic form of satisficing~\cite{simon1955satisficing} rather than silently converting every tradeoff into one weighted objective.

Human-facing examples add a further distinction between measurement error and substantive heterogeneity. Variability across individuals or conditions need not be mere noise; in neuroscience, it has been interpreted as an adaptive property of complex organization, while current measurement approaches must still separate biological variation from methodological uncertainty~\cite{forkel2026neurovariability}. The connection used here is methodological, not physiological: pooled probability localizations require an explicit transportability assumption, and different internal configurations or policies may be functionally equivalent for a declared outcome without being structurally identical.

The paper makes four contributions. First, it formulates an integrated dynamic strategy through coupled allocation, internal-state, return, and cost localizations with a causal regime convention. Second, it gives graph- and design-relative identifiability conditions and distinguishes the SSRC reduction of repeated tensor words from additional relations induced by the simplex and by the observed design. Third, it develops acceptability certificates from finite-horizon error propagation and separates nominal acceptability, certified acceptability, model-relative near-optimality, and finite-class certified $\epsilon$-optimality. Fourth, it introduces a regime-indexed stochastic feedback subclass and evaluates exact, operator-supplied, and multinomially noisy instances in reproducible synthetic laboratories.

The intended scope is finite-state processes and finite localizations of compact state spaces. The results do not constitute a general solution to stochastic control, guarantee identification beyond the excited domain, or calibrate physiological or commercial behavior. Potentially cyclic operation motivates retaining time, regime, and cost information beyond terminal occupancy; results for circulation in homogeneous Markov chains~\cite{jia2016} motivate such observables but are not transferred directly to adaptive switched policies. Section~II develops the representation, identification, certification, and feedback constructions. Section~III states the combined computational workflow, Sec.~IV clarifies switching and enrichment, Sec.~V reports the numerical laboratories, and Sec.~VI discusses implications and limitations.

\section{Coupled stochastic descriptions}
Write $\D^{d-1}=\{z\in\R^d:z\ge0,\ones^\top z=1\}$. At decision time $t$,
\begin{equation}
u_t\in\D^{m-1},\quad p_t\in\D^{n-1},\quad
q_t\in\D^{n_q-1},\quad v_t\in\D^{s-1}.
\end{equation}
These describe input composition, internal-state localization, return categories, and cost categories, respectively. A composition does not determine the total available resource: a budget $b_t$ must be supplied separately whenever absolute scale affects evolution. Distinct resources can require a product of allocation simplexes rather than one shared simplex.

Return and cost categories are measurable, disjoint, and exhaustive within their respective observation spaces. If $Y_t$ and $c_t$ denote return and cost observations, then $q_t[a]=\Pr(Y_t\in R_a)$ and $v_t[b]=\Pr(c_t\in C_b)$. Joint categories can encode, for example, both depletion and compensating resource consumption without assigning a common monetary value.

\begin{definition}
An integrated dynamic strategy is a sequence of admissible allocation rules based on available information, designed to direct the joint evolution of internal localization, returns, and costs under resource and operational constraints.
\end{definition}

Let $\mathcal I_t^-$ denote the information available immediately before the action at time $t$. Define the controller-accessible localization $\bar p_t$ by $\bar p_t=p_t$ when the internal localization is observed and by $\bar p_t=\Phi_t(\mathcal I_t^-)$ when it must be produced by a declared filter or estimator $\Phi_t$. The within-period causal order is
\begin{equation}
\mathcal I_t^-\longrightarrow(\bar p_t,\zeta_t)
\longrightarrow\nu_t\longrightarrow u_t
\longrightarrow(q_t,v_t,p_{t+1}).
\label{eq:causal-order}
\end{equation}
Thus the current regime is either supplied exogenously or selected before the action by $\nu_t=g(\bar p_t,\zeta_t)$, and an allocation rule then chooses $u_t=\pi_t(\mathcal I_t^-,\nu_t)\in\mathcal U_t(b_t)$. The primary prototype is the switched additive stochastic model
\begin{align}
p_{t+1}&=\alpha_p W_{p,\nu_t}p_t+(1-\alpha_p)V_{p,\nu_t}u_t,\\
q_t&=\alpha_q W_{q,\nu_t}p_t+(1-\alpha_q)V_{q,\nu_t}u_t,\\
v_t&=\alpha_v W_{v,\nu_t}p_t+(1-\alpha_v)V_{v,\nu_t}u_t.
\label{eq:switched}
\end{align}
Here $\nu_t$ denotes the operating regime, not enrichment depth; enrichment is indexed separately by a superscript $(\ell)$ when needed. All $W,V$ blocks are column-stochastic, with dimensions $n\times n,n\times m$ for internal dynamics and $n_q\times n,n_q\times m$ or $s\times n,s\times m$ for outputs. Mixing coefficients lie in $[0,1]$. The output timing in Eq.~\eqref{eq:switched} is contemporaneous with $(p_t,u_t)$; other timing conventions must be specified when comparing laboratories.

A prescribed regime sequence gives a time-varying affine system. In the fully observed case $\bar p_t=p_t$, a piecewise-constant pre-action rule $\nu_t=g(p_t,\zeta_t)$ partitions the state--context domain; substitution in Eq.~\eqref{eq:switched} then yields a piecewise-affine state-dependent map and therefore a nonlinear global evolution~\cite{hamilton1989regime,liberzon2003switching}, although every active local model remains linear in its stochastic embedding. If the action is intended to affect regime selection, causality is preserved by assigning the result to the next period, for example $\nu_{t+1}=g(\bar p_t,u_t,\zeta_t)$, rather than using $u_t$ to define $\nu_t$ while simultaneously using $\nu_t$ to define $u_t$. Under partial observation, the global dynamics also include the declared filtering recursion and need not be piecewise affine in the latent state alone. This is the intended moderate use of switching: the regime family can distinguish nonlinear operating regions without attributing nonlinearity to a fixed, externally prescribed schedule. The prototypes use prescribed, observed regimes common to the represented population. Hidden, uncertain, or heterogeneous regimes require joint regime inference, posterior averaging, or robust constraints; a policy may not condition on an unavailable realized regime. Bayesian priors may constrain blocks or mixing coefficients, while the reported experiments fix the coefficients and identify the blocks from collective output counts.

For fixed regime and input, the internal transition is equivalently
\begin{equation}
P_\nu(u)=\alpha_pW_{p,\nu}+(1-\alpha_p)(V_{p,\nu}u)\ones^\top.
\end{equation}
This is a specific restricted controlled kernel: the input contribution is independent of the source state. More general state--input interactions can be introduced through joint state--input kernels or enriched models. When the internal localization is not observed, propagation may still use the latent model state $p_t$, but implementable policies must use the accessible localization $\bar p_t$ and a declared filtering recursion; replacing $\bar p_t$ by the unavailable $p_t$ would be an oracle policy.

\subsection{Regime-wise SSRC identification}
Let $y_{p,t}=p_{t+1}$, $y_{q,t}=q_t$, and $y_{v,t}=v_t$. With observed historical allocations $u_t$, use an output-specific stochastic embedding
\begin{equation}
z_{j,t}=\begin{bmatrix}\alpha_jp_t\\(1-\alpha_j)u_t\end{bmatrix},
\qquad y_{j,t}=M_{j,\nu_t}z_{j,t},
\end{equation}
where $M_{j,\nu}=[W_{j,\nu}\ V_{j,\nu}]$. This gives a degree-one block realization of the SSRC identification architecture in Ref.~\onlinecite{banegas2025ssrc}, adapted to the probability-localization setting of Ref.~\onlinecite{vides2026localization}. If a common arbitrary $\beta\in(0,1)$ is used instead, its fitted blocks must be scaled by $\alpha_j/\beta$ and $(1-\alpha_j)/(1-\beta)$; the resulting operator generally has prescribed block column sums rather than unit column sums. Taking $\beta_j=\alpha_j$ avoids this distinction in the first prototype.

The prototype adapts the supplied pairwise SSRC routine: construct the design only on graph-admissible entries, solve the augmented nonnegative normal-equation system with column-sum equations, and normalize positive columns. The implementation uses a NumPy active-set NNLS solver~\cite{lawson1974solving} and records conditioning and pre-normalization errors. It does not identify this procedure with exact equality-constrained least squares, nor claim unconditional numerical stability.

The relevant rank test must respect both the graph and the stochastic constraints. For an output dimension $r$, let
\begin{equation}
\mathcal M_G=\{M\in\mathbb S_{r,d}(\mathbb R):M_{ij}=0\text{ when }(i,j)\notin G\}
\end{equation}
and let $\operatorname{par}(\mathcal M_G)$ be its parallel space of graph-supported matrices $D$ satisfying $\ones^\top D=0$. If $Z=[z_1,\ldots,z_N]$ and the columns of $B_G$ form a basis for $\operatorname{vec}(\operatorname{par}(\mathcal M_G))$, define the restricted design
\begin{equation}
\mathcal D_G(Z)=(Z^\top\otimes I_r)B_G.
\label{eq:restricted-design}
\end{equation}

\begin{proposition}[Graph- and design-relative identifiability]
\label{prop:restricted-identifiability}
The observation map $M\mapsto MZ$ is injective on $\mathcal M_G$ if
\begin{equation}
\ker(Z^\top\otimes I_r)\cap
\operatorname{vec}(\operatorname{par}(\mathcal M_G))=\{0\},
\label{eq:restricted-identifiability}
\end{equation}
equivalently if $\operatorname{rank}\mathcal D_G(Z)=\dim\operatorname{par}(\mathcal M_G)$. Conversely, if Eq.~\eqref{eq:restricted-identifiability} fails, every $M$ in the relative interior of $\mathcal M_G$ is observationally equivalent on $Z$ to another member of $\mathcal M_G$.
\end{proposition}
\begin{proof}
For $M,M'\in\mathcal M_G$, equality $MZ=M'Z$ is equivalent to $(Z^\top\otimes I_r)\operatorname{vec}(M-M')=0$, while $M-M'\in\operatorname{par}(\mathcal M_G)$. Hence Eq.~\eqref{eq:restricted-identifiability} implies $M=M'$. Conversely, a nonzero direction in the intersection can be added with sufficiently small positive or negative coefficient to any relative-interior point without violating its support, nonnegativity, or column sums.
\end{proof}

Accordingly, the reported constrained rank, nullity, and condition number are computed from $\mathcal D_G(Z)$, not from the unconstrained graph-entry design alone. Rank deficiency then has a precise meaning: the observations do not distinguish every locally feasible stochastic coefficient direction. Boundary points may still be unique through active nonnegativity constraints, so the restricted rank test is sufficient throughout the feasible class and necessary at its relative-interior points.

Separate recovery of both blocks requires $0<\alpha_j<1$, sufficient excitation, and structural identifiability; at an endpoint one block is unobserved. In a dense model with $0<\alpha<1$, the transformations
\begin{equation}
W\mapsto W+d\ones^\top,\qquad
V\mapsto V-\frac{\alpha}{1-\alpha}d\ones^\top,
\end{equation}
with $\ones^\top d=0$, preserve predictions whenever nonnegativity remains valid. Graph restrictions can remove this ambiguity; data coverage and constrained design rank must be checked. A prior can select among observationally equivalent blocks without making them data-identifiable.

\begin{proposition}[Common-schedule forgetting]
\label{prop:common-schedule-forgetting}
For two trajectories of Eq.~\eqref{eq:switched} with the same regime sequence and inputs,
\begin{equation}
\|p_t-\widetilde p_t\|_1\le\alpha_p^t\|p_0-\widetilde p_0\|_1.
\end{equation}
Moreover, on the zero-mass subspace, the restriction of $P_\nu(u)$ equals the restriction of $\alpha_pW_{p,\nu}$ and is independent of constant $u$.
\end{proposition}
\begin{proof}
The common input terms cancel, and column-stochastic matrices are nonexpansive in $\ell^1$. Iterate this inequality. For $\ones^\top z=0$, the rank-one input term satisfies $(V_{p,\nu}u)\ones^\top z=0$.
\end{proof}
Strict forgetting requires $\alpha_p<1$. Different state-dependent switches or feedback inputs invalidate the cancellation unless separately controlled. Thus within this additive prototype, changing a constant allocation can change output localization and the stationary law without changing fixed-regime zero-mass modes. This dynamical observation is distinct from stability of the identification algorithm.

\subsection{Acceptability under residual uncertainty}
The primary decision need not maximize target probability~\cite{simon1955satisficing}. Let $E_H$ be a specified achievement event and $B_H$ an unsustainable-cost event over horizon $H$. These can represent hitting, sustained performance, cumulative budget exceedance, or persistent adverse operation. For available information $\mathcal I$, write
\begin{align}
R_H(\pi)&=\Pr(E_H\mid\mathcal I,\pi),\\
C_H(\pi)&=\Pr(B_H\mid\mathcal I,\pi).
\end{align}
Given an implementable policy class $\Pi(\mathcal I)$, an aspiration threshold $\alpha$, and a tolerable risk $\beta$, define
\begin{equation}
\mathcal A=\{\pi\in\Pi(\mathcal I):R_H(\pi)\ge\alpha,
\ C_H(\pi)\le\beta\}.
\label{eq:objective}
\end{equation}
The primary task is to find an element of $\mathcal A$. Cost, complexity, or implementation preferences may subsequently distinguish acceptable strategies. Thresholds are context-dependent and are not silently relaxed when no strategy can be certified. The separate constraints imply only
\begin{equation}
\Pr(E_H\cap B_H^c\mid\mathcal I,\pi)\ge\max(0,\alpha-\beta).
\end{equation}
A stronger joint requirement must be stated separately. Marginal laws $q_t,v_t$ alone do not determine these path probabilities. A joint process model, with event flags or accumulators where needed, is required.

More generally, the achievement and cost events need not be singular or evaluated over the full horizon. Given finite index sets $I,J$, coordinate subsets $\mathcal C_i,\mathcal C_j$, and time subsets $T_i,T_j\subseteq\{1,\dots,H\}$, define achievement probabilities $R_i(\pi)=\Pr(E_i\mid\mathcal I,\pi)$ and cost-event probabilities $C_j(\pi)=\Pr(B_j\mid\mathcal I,\pi)$ for events restricted to $(T_i,\mathcal C_i)$ and $(T_j,\mathcal C_j)$ respectively. The generalized acceptability region is
\begin{equation}
\mathcal A_{\mathcal T} = \{\pi\in\Pi(\mathcal I): R_i(\pi)\ge\alpha_i\ \forall i\in I,\ C_j(\pi)\le\beta_j\ \forall j\in J\},
\end{equation}
recovering Eq.~\eqref{eq:objective} as the singleton case $I=J=\{1\}$, $T_1=\{H\}$. Achievement and avoidance restrictions remain separate rather than combined into a single weighted score, so that a shortfall in one cannot be offset by a surplus in another. Proposition~\ref{prop:acceptability-margins} applies componentwise to each $(R_i,\delta_{R_i})$ and $(C_j,\delta_{C_j})$ pair without change of form.

Distinct operating circumstances may justify a finite family of declared decision postures $h\in\mathcal H$, with posture-specific thresholds $(\alpha^{(h)},\beta^{(h)})$ and acceptability regions $\mathcal A^{(h)}$. The active posture must be selected from pre-action information by an auditable rule before policy comparison. It cannot be chosen after optimization merely because another posture yields a desirable candidate. This construction represents context-dependent risk tolerance without silently relaxing an unsuccessful certificate.

\begin{definition}[Decision sufficiency relative to acceptability]
\label{def:decision-sufficiency}
For fixed information, policy class, events, horizon, and thresholds, a representation is sufficient for acceptability classification if it preserves membership in $\mathcal A$. A weaker selection requirement is to certify at least one implementable policy in $\mathcal A$, even if other policies remain unresolved.
\end{definition}

This target-relative requirement does not demand exact reconstruction of every internal transition. Consider estimated event probabilities $\widehat R_H,\widehat C_H$. The nominal acceptability region under the estimated model is
\begin{equation}
\widehat{\mathcal A}_0=\{\pi\in\Pi(\mathcal I):
\widehat R_H(\pi)\ge\alpha,\ \widehat C_H(\pi)\le\beta\}.
\label{eq:nominal-region}
\end{equation}
Membership in $\widehat{\mathcal A}_0$ is a model-relative statement and is not yet a certificate for membership in $\mathcal A$. Suppose instead that justified, possibly policy-dependent bounds satisfy
\begin{equation}
|R_H(\pi)-\widehat R_H(\pi)|\le\delta_R(\pi),\qquad
|C_H(\pi)-\widehat C_H(\pi)|\le\delta_C(\pi).
\label{eq:event-bounds}
\end{equation}
The corresponding certified region is
\begin{multline}
\widehat{\mathcal A}_{\delta}=\{\pi\in\Pi(\mathcal I):
\widehat R_H(\pi)-\delta_R(\pi)\ge\alpha,\\
\widehat C_H(\pi)+\delta_C(\pi)\le\beta\}.
\label{eq:certified-region}
\end{multline}
Thus $\widehat{\mathcal A}_{\delta}\subseteq\widehat{\mathcal A}_0$, and valid bounds imply $\widehat{\mathcal A}_{\delta}\subseteq\mathcal A$ on their coverage event.

The following result supplies one conservative source of such margins for the pointwise criteria used below. Let $M_{j,\nu}=[W_{j,\nu}\ V_{j,\nu}]$ and $\widehat M_{j,\nu}$ denote true and estimated column-stochastic operators, and assume
\begin{equation}
\|\widehat M_{j,\nu}-M_{j,\nu}\|_1\le\varepsilon_{j,\nu},
\qquad j\in\{p,q,v\}.
\label{eq:operator-bounds}
\end{equation}
\begin{proposition}[Finite-horizon open-loop propagation]
\label{prop:open-loop-propagation}
Consider true and estimated trajectories with a common prescribed regime sequence and common open-loop inputs. If $e_t=\|\widehat p_t-p_t\|_1$, then
\begin{align}
e_{t+1}&\le\alpha_p e_t+\varepsilon_{p,\nu_t},\\
\|\widehat q_t-q_t\|_1&\le\alpha_q e_t+\varepsilon_{q,\nu_t},\\
\|\widehat v_t-v_t\|_1&\le\alpha_v e_t+\varepsilon_{v,\nu_t}.
\end{align}
Consequently,
\begin{equation}
e_t\le\alpha_p^t e_0+
\sum_{k=0}^{t-1}\alpha_p^{t-1-k}\varepsilon_{p,\nu_k}.
\label{eq:state-error-sum}
\end{equation}
For any return category $a$, cost category $b$, terminal index $T$, and evaluation set $S$,
\begin{align}
|\widehat q_T[a]-q_T[a]|&\le
\frac12(\alpha_qe_T+\varepsilon_{q,\nu_T})=: \delta_R,\\
\left|\max_{t\in S}\widehat v_t[b]-\max_{t\in S}v_t[b]\right|&\le
\frac12\max_{t\in S}(\alpha_ve_t+\varepsilon_{v,\nu_t})=: \delta_C.
\end{align}
\end{proposition}
\begin{proof}
For the state recursion, add and subtract
\begin{equation*}
\widehat M_{p,\nu_t}
[\alpha_pp_t^\top,(1-\alpha_p)u_t^\top]^\top.
\end{equation*}
Column-stochasticity makes $\widehat M_{p,\nu_t}$ nonexpansive in $\ell^1$, the common input components cancel, and Eq.~\eqref{eq:operator-bounds} bounds the remaining model discrepancy. The output inequalities follow identically. Iteration gives Eq.~\eqref{eq:state-error-sum}. Finally, the difference of two probability vectors has zero total mass, so every coordinate difference is bounded by one half of its $\ell^1$ norm; the maximum functional is nonexpansive in the uniform norm.
\end{proof}
The result is conditional on common inputs and regimes. Feedback, state-dependent regime disagreement, hitting events, and cumulative events require separate propagation arguments. Large coefficient-level bounds may also be conservative when observationally equivalent operators predict similarly on the excited domain. Tighter bounds valid only on a declared finite policy class remain sufficient for classification, but their scope must be reported.

\begin{proposition}[Acceptability margins]
\label{prop:acceptability-margins}
A policy is certified acceptable if
\begin{equation}
\widehat R_H(\pi)-\delta_R(\pi)\ge\alpha,\qquad
\widehat C_H(\pi)+\delta_C(\pi)\le\beta.
\end{equation}
It is certified unacceptable if $\widehat R_H(\pi)+\delta_R(\pi)<\alpha$ or $\widehat C_H(\pi)-\delta_C(\pi)>\beta$. Otherwise these bounds leave its status unresolved.
\end{proposition}
\begin{proof}
The lower achievement bound and upper risk bound imply both defining inequalities of $\mathcal A$. Conversely, either stated rejection condition forces violation of at least one inequality. Equality at a threshold remains admissible.
\end{proof}
This elementary implication is conditional on valid error bounds, not a method for obtaining them. Statistical bounds certify the conclusion on their coverage event. Coverage uncertainty is distinct from operational risk $\beta$; data-dependent policy selection requires simultaneous coverage or independent validation of the selected policy. Representation and estimation errors must both be addressed.

Residual uncertainty is allowed. No universal entropy-minimization objective is imposed on $u,p,q,v$. Entropy alone does not distinguish desirable from undesirable categories: permuting their probabilities preserves entropy while changing target probability. Enrichment is useful here when it resolves an acceptability margin or enables a less costly acceptable policy, not merely when it concentrates a distribution.

\subsection{Negotiation regions and near-optimal reference sets}
Optimization remains useful as a reference. Let $z=(\pi,\xi)$ combine an implementable strategy and negotiable terms, such as resource budgets, participation commitments, and delivery windows~\cite{nash1950bargaining}. Let $\mathcal F$ encode nonnegotiable feasibility conditions, and let $\mathcal A_i$ be party $i$'s acceptance region, including achievement and cost-risk thresholds where appropriate. Define
\begin{equation}
\mathcal N=\mathcal F\cap\bigcap_{i=1}^{M}\mathcal A_i.
\label{eq:negotiation}
\end{equation}
For a declared common criterion $J$, a reference and near-optimal negotiation set are
\begin{align}
J^\star&=\sup_{z\in\mathcal N}J(z),\\
\mathcal N_\epsilon&=\{z\in\mathcal N:J(z)\ge J^\star-\epsilon\}.
\label{eq:nearoptimal}
\end{align}
Without an agreed scalar criterion, parties can instead compare multiple objectives and concessions~\cite{miettinen1999nonlinear}. Near-optimality does not itself relax feasibility or risk constraints. A neighborhood in objective value need not be a geometric ball in allocation space: nearby allocations may have different consequences, while distant allocations may have similar values. A surrogate estimates these regions; it does not establish real-world feasibility without appropriate validation.

For illustration, participation requirements $x_1+x_2\le0.50$ and $0.10\le x_2\le0.15$ permit the boundary segment $(x_1,x_2)=(0.38-\delta,0.12+\delta)$ for $-0.02\le\delta\le0.03$. The point $(0.38,0.12)$ is only a stipulated optimum until $J$ is specified. Value loss along this segment must be evaluated, not inferred from compositional distance.

Supplier participation can be negotiated over an operating window rather than at each order. For allocated volume $b_t$,
\begin{equation}
\bar u_{i,H}=\frac{\sum_{t=0}^{H-1}b_tu_{i,t}}{\sum_{t=0}^{H-1}b_t},
\qquad \sum_t b_t>0.
\end{equation}
Expected randomized participation, realized volume share, and expenditure share are distinct. Personal resource negotiation similarly concerns task-specific preparation under time and recovery constraints, rather than maximal physical performance. Preparation and successful task completion remain different events.

When only the fitted criterion $\widehat J$ is used for ranking, define
\begin{equation}
\widehat{\mathcal N}^{\rm model}_{\epsilon,\delta}
=\left\{z\in\widehat{\mathcal A}_\delta:
\widehat J(z)\ge
\sup_{z'\in\widehat{\mathcal A}_\delta}\widehat J(z')-\epsilon\right\}.
\label{eq:model-nearoptimal}
\end{equation}
Every member is certified acceptable, but the near-optimal ranking is model-relative and compares only policies already in $\widehat{\mathcal A}_\delta$.

\begin{proposition}[Finite-class $\epsilon$-optimality certificate]
\label{prop:finite-epsilon-certificate}
Consider a finite structurally feasible candidate class with simultaneous objective intervals $[L_J(z),U_J(z)]$. Let $\mathcal P$ contain all candidates not certified unacceptable and set
\begin{equation}
U^\star=\max_{z\in\mathcal P}U_J(z).
\end{equation}
On the simultaneous-coverage event, any certified-acceptable candidate satisfying
\begin{equation}
U^\star-L_J(z)\le\epsilon
\label{eq:strong-nearoptimal}
\end{equation}
is $\epsilon$-optimal relative to all truly acceptable candidates in the finite class.
\end{proposition}
\begin{proof}
Every truly acceptable candidate remains in $\mathcal P$, so its true objective is at most $U^\star$. The candidate in Eq.~\eqref{eq:strong-nearoptimal} has true objective at least $L_J(z)\); hence its loss relative to the best truly acceptable candidate is at most $U^\star-L_J(z)\le\epsilon$.
\end{proof}
For $J=R_H$, the objective intervals are $[L_R,U_R]$. The certificate is deliberately conservative because unresolved but truly unacceptable candidates may raise $U^\star$. It is not an optimality claim over policies outside the declared finite class.

\subsection{Stochastic feedback closure and certified feedback feasibility}
The primary prototype in Eq.~\eqref{eq:switched} allows $u_t$ to be selected by a general implementable policy. This subsection identifies a distinguished subclass in which the same allocation variable is generated by stochastic feedback from the controller-accessible localization $\bar p_t$. In the fully observed case $\bar p_t=p_t$; under partial observation, $\bar p_t$ must be supplied by the declared estimator in Eq.~\eqref{eq:causal-order}.

Let $\eth_{s,2}$ denote the reduced stochastic embedding of degree at most two introduced in Ref.~\cite{vides2026localization} (Eqs.~16--20 there), with $s=(s_1,s_2,s_3)\in\{0,1\}^3$ selecting the linear block, the quadratic block, and a constant term respectively. The SSRC reduction matrix $R_{s,p}(n)$ aggregates permutation-equivalent coordinates of $x^{\otimes k}$ for every selected order $k\le p$; thus it removes the repeated tensor words before identification. In general, the resulting ambient dimension is
\begin{equation}
d_{s,p}=s_{p+1}+\sum_{k=1}^{p}s_k{n+k-1\choose k},
\label{eq:reduced-embedding-dimension}
\end{equation}
and for $s=(1,1,0)$ it is $d_s=n+{n+1\choose2}$. Write $\eth_2:=\eth_{s,2}$ for this typical choice.

This tensor reduction does not remove algebraic relations created by restricting several degree blocks to the simplex. For example, $\sum_jp_j=1$ implies
\begin{equation}
p_i=p_i^2+\sum_{j\ne i}p_ip_j,
\label{eq:simplex-cross-degree}
\end{equation}
with the coefficients adjusted to the multiplicity scaling used by $R_{s,2}(n)$. More generally, a monomial of degree $k<p$ can be homogenized on the simplex by multiplication by $(\sum_i p_i)^{p-k}=1$. Consequently, if the complete degree-$p$ block is active, the functional span of all selected blocks on $\Delta^{n-1}$ has dimension at most ${n+p-1\choose p}$ even when the reduced ambient dimension in Eq.~\eqref{eq:reduced-embedding-dimension} is larger. These are domain-induced cross-degree relations, not tensor-product redundancies.

For each regime $\nu$, let $K_\nu\in\mathbb S_{m,d_s}(\mathbb R)$ be column-stochastic, and define
\begin{equation}
u_t = K_{\nu_t}\eth_2(\bar p_t).
\label{eq:feedback}
\end{equation}
Substituting Eq.~\eqref{eq:feedback} in Eq.~\eqref{eq:switched} gives the closed-loop switched model
\begin{equation}
p_{t+1} = \alpha_p W_{p,\nu_t} p_t + (1-\alpha_p) V_{p,\nu_t} K_{\nu_t}\eth_2(\bar p_t),
\label{eq:closedloop}
\end{equation}
and analogously for $q_t,v_t$, since $u_t$ enters through the same $V_{j,\nu}$ block in each equation.

\begin{proposition}[Stochastic closure]
If $\bar p_t\in\Delta^{n-1}$, $\eth_2:\Delta^{n-1}\to\Delta^{d_s-1}$ is stochastic, and $K_\nu\in\mathbb S_{m,d_s}(\mathbb R)$ for every $\nu$, then $u_t\in\Delta^{m-1}$. In the fully observed case $\bar p_t=p_t$, the closed-loop recursion above is a well-posed stochastic recursion on $\Delta^{n-1}$, with $q_t,v_t$ likewise remaining in their respective simplices. Under partial observation, the same conclusion holds for the plant state when the recursion is coupled to a declared filter that returns $\bar p_t\in\Delta^{n-1}$.
\end{proposition}
\begin{proof}
$\ones^\top u_t = \ones^\top K_{\nu_t}\eth_2(\bar p_t) = \ones^\top \eth_2(\bar p_t) = 1$, and $u_t\ge 0$ since $K_{\nu_t}$ and $\eth_2(\bar p_t)$ are nonnegative. The claim follows because $p_{t+1}$ is a convex combination of stochastic vectors under column-stochastic maps. In the partially observed case, the declared filtering recursion supplies the additional state needed for causal propagation.
\end{proof}

This substitution invalidates the cancellation used in Proposition~\ref{prop:common-schedule-forgetting}, since $u_t$ now depends on the trajectory through $\bar p_t$ rather than being common and externally fixed; this is exactly the feedback case excluded from that cancellation.

The substitution identifies a policy subclass
\begin{multline}
\Pi_K^{(s)}(\mathcal I) := \{\pi_K : u_t = K_{\nu_t}\eth_2(\bar p_t),\\
K_\nu \in \mathbb S_{m,d_s}(\mathbb R)\ \forall \nu\} \subset \Pi(\mathcal I),
\end{multline}
of regime-indexed, graph-structured stochastic quadratic feedback policies. $\Pi_K^{(s)}$ is not asserted to exhaust $\Pi(\mathcal I)$; it is an additional implementable class that can be tested for membership in $\mathcal A$ alongside the constant-allocation candidates already considered.

For any policy with the event-probability bounds in Eq.~\eqref{eq:event-bounds}, define its conservative acceptability slack by
\begin{equation}
\rho(\pi)=\min\left\{
\widehat R_H(\pi)-\delta_R(\pi)-\alpha,\;
\beta-\widehat C_H(\pi)-\delta_C(\pi)
\right\}.
\label{eq:certified-slack}
\end{equation}
Thus $\rho(\pi)\ge0$ is equivalent to the acceptance certificate in Proposition~\ref{prop:acceptability-margins}. For a declared feedback family $\mathcal C_K\subseteq\Pi_K^{(s)}(\mathcal I)$, write
\begin{equation}
\rho_{\mathcal C_K}^{\star}=\sup_{\pi_K\in\mathcal C_K}\rho(\pi_K),
\label{eq:feedback-feasibility-value}
\end{equation}
using a maximum when $\mathcal C_K$ is finite.

\begin{definition}[$\eta$-certified feedback feasibility]
For fixed information, events, horizon, thresholds, error bounds, and a declared feedback family $\mathcal C_K\subseteq\Pi_K^{(s)}(\mathcal I)$, the family is $\eta$-certifiably feasible, for $\eta\ge0$, if $\rho_{\mathcal C_K}^{\star}\ge\eta$; equivalently, some $\pi_K\in\mathcal C_K$ satisfies $\rho(\pi_K)\ge\eta$.
\end{definition}

Here $\eta$ is a required nonnegative robustness reserve in probability units, not a permitted constraint violation. The case $\eta=0$ recovers ordinary certified acceptability. This notion is deliberately family-relative: failure for $\mathcal C_K$ says nothing about $\Pi(\mathcal I)$ or even about untested members of $\Pi_K^{(s)}(\mathcal I)$, consistent with the target-relative treatment of decision sufficiency in Definition~\ref{def:decision-sufficiency}. A vector of constraint-specific reserves can replace the scalar minimum when return and cost margins should not be treated symmetrically. This definition concerns certified feasibility, not state- or distribution-reachability controllability.

\textbf{Identification of $K_\nu$.} Two complementary routes are available, and the framework does not privilege one over the other.

\textit{(i) Estimation from data.} Given observed pairs $(p_t,u_t)$ with $\nu_t=\nu$, set $z_t=\eth_2(p_t)$, $\mathcal T_\nu=\{t:\nu_t=\nu\}$, $N_\nu=|\mathcal T_\nu|$, and form the data matrices
\begin{equation}
U_\nu=[u_t]_{t\in\mathcal T_\nu},\qquad
Z_\nu=[z_t]_{t\in\mathcal T_\nu}.
\end{equation}
Let $\mathcal K_\nu=(\operatorname{span}\mathcal B_{S,\nu}(m,d_s))\cap\mathbb S_{m,d_s}(\mathbb R)$ denote the column-stochastic matrices supported by the regime-specific relational graph $G_{S,\nu}$. The graph-constrained least-squares estimate is
\begin{align}
\widehat K_\nu&\in\arg\min_{K\in\mathcal K_\nu}
\frac{1}{N_\nu}\|U_\nu-KZ_\nu\|_F^2 \notag\\
&=\arg\min_{K\in\mathcal K_\nu}
\frac{1}{N_\nu}\sum_{t\in\mathcal T_\nu}\|u_t-Kz_t\|_2^2.
\label{eq:feedback-identification}
\end{align}
Thus residuals are penalized sample by sample. The norm of their sum is not used, since positive and negative residuals from different observations could cancel. The constrained problem can be assembled with the same graph-supported design used in Sec.~II.A, with nonnegativity and unit column sums imposed on $K$. Unequal-precision observations can instead use a declared weighted sum of squared residuals; allocation counts may be fitted by an appropriate multinomial likelihood.

No equality constraint is imposed across regimes: $\widehat K_\nu$ is estimated independently for each $\nu$, and a complete $G_{S,\nu}$ recovers the dense feasible class. Unique recovery requires the sampled reduced embedding to distinguish all feasible coefficient directions. In particular, the specialization of Proposition~\ref{prop:restricted-identifiability} is
\begin{equation}
DZ_\nu=0,\quad D\in\operatorname{par}(\mathcal K_\nu)
\quad\Longrightarrow\quad D=0,
\label{eq:feedback-identifiability}
\end{equation}
where $\operatorname{par}(\mathcal K_\nu)$ is the linear space of support-admissible, zero-column-sum differences. If this condition fails, distinct feasible controllers can be observationally equivalent on the training design; held-out functional prediction may still be accurate on a declared domain. Even arbitrarily rich sampling of the simplex cannot distinguish a direction that annihilates the entire set $\eth_2(\Delta^{n-1})$. For $n=4$ and $s=(1,1,0)$, tensor reduction has already reduced the embedding from $4+4^2=20$ to $4+10=14$ coordinates, while the four independent relations in Eq.~\eqref{eq:simplex-cross-degree} leave a $10$-dimensional functional span. Thus a dense $K_\nu\in\mathbb S_{m,14}(\mathbb R)$ is generally identifiable only up to its action on that embedded domain unless further structure selects a realization. Approximate persistence across regimes, when observed, is therefore an empirical finding rather than an identifying restriction.

When a lower-dimensional parametrization $K_\nu=\mathcal L_\nu(\theta_\nu)$ is declared, estimation should be performed in that parametrization rather than in an ambient realization containing domain-induced null directions. In the personal laboratory below, $\theta_\nu=A_\nu\in\mathbb S_{m,n}(\mathbb R)$ and the canonical lifting satisfies $\mathcal L_\nu(A_\nu)\eth_2(p)=A_\nu p$. Consequently, with $P_\nu=[p_t]_{t\in\mathcal T_\nu}$, the implemented exact-data fit is the special case
\begin{equation}
\widehat A_\nu\in\arg\min_{A\in\mathbb S_{m,n}(\mathbb R)}
\frac{1}{N_\nu}\|U_\nu-AP_\nu\|_F^2,
\qquad \widehat K_\nu=\mathcal L_\nu(\widehat A_\nu).
\label{eq:minimal-feedback-identification}
\end{equation}
The full-row-rank design used there makes the unconstrained least-squares solution unique; because the data are generated exactly by a feasible stochastic $A_\nu$, it coincides with the constrained solution up to floating-point cleanup.

\textit{(ii) Supplied candidates.} $K_\nu$ may equally be elicited from a Bayesian prior over column-stochastic matrices or supplied from declared operator experience. The certification stage (Proposition~\ref{prop:acceptability-margins}, instantiated with $\Pi_K^{(s)}$) treats any such candidate identically to an estimated one, requiring only that $K_\nu$ be column-stochastic and, when a structural restriction is declared, graph-admissible.

Residual uncertainty from route (i) and lack of certified feasibility within a declared feedback family are distinct conclusions and should not be conflated when reporting $\rho(\pi_K)$ or $\rho_{\mathcal C_K}^{\star}$.

More generally, two controllers $K_\nu,K_\nu'$ that are far apart as column-stochastic matrices may induce similar, or jointly admissible, localization trajectories on a declared subset $\mathcal T$ of times and coordinates. This disassociation between distance in policy and distance in outcome, already noted for constant allocations in Sec.~II.C, extends naturally to $\Pi_K^{(s)}$; formalizing an outcome-level equivalence over $\mathcal T$ and a corresponding minimal-complexity selection criterion among certified candidates is left as an open direction.

\section{Identification and synthesis algorithm}
Algorithm~\ref{alg:ids} joins the two computational stages without conflating their purposes. The first estimates regime-local stochastic operators and checks whether the excited data support prediction at the resolution needed for decision classification. The second propagates candidate strategies, applies conservative return--cost margins, and constructs the negotiable set. Consequently, failure to certify a strategy is an admissible output rather than an instruction to relax the declared thresholds.

\begin{algorithm2e*}[t]
\small
\caption{Regime-wise identification and negotiable IDS synthesis}
\label{alg:ids}
\KwIn{Data $\mathcal D=\{(p_t,u_t,p_{t+1},q_t,v_t,\nu_t)\}$; supports $G_{j,\nu}$; weights $\alpha_j$; policy class $\Pi$; horizon $H$; declared posture $h$ with thresholds $(r_{\min}^{(h)},c_{\max}^{(h)})$; tolerance $\epsilon$.}
\KwOut{Identified operators and diagnostics; nominal, certified, unresolved, model-relative near-optimal, and, when available, finite-class $\epsilon$-optimal-certified regions; an implementable reference.}
\tcp{Stage I: identify and validate regime-local dynamics}
\ForEach{regime $\nu$ and output $j\in\{p,q,v\}$}{
  Form $z_{j,t}=[\alpha_jp_t^\top,(1-\alpha_j)u_t^\top]^\top$ for samples with $\nu_t=\nu$\;
  Fit $\widehat M_{j,\nu}=[\widehat W_{j,\nu}\ \widehat V_{j,\nu}]$ by graph-supported nonnegative least squares with unit column sums\;
  Record held-out error, constrained design rank, and conditioning\;
}
\If{validation is inadequate near a decision boundary}{
  \Return{unresolved representation; revise data, regimes, graph, or enrichment}\;
}
\tcp{Stage II: synthesize a negotiable IDS}
\ForEach{implementable $\pi\in\Pi$}{
  Propagate $\widehat p_t,\widehat q_t,\widehat v_t$ and estimate $(\widehat R_H,\widehat C_H,\delta_R,\delta_C)$\;
}
$\widehat{\mathcal A}_0^{(h)}\leftarrow\{\pi:\widehat R_H(\pi)\ge r_{\min}^{(h)},\ \widehat C_H(\pi)\le c_{\max}^{(h)}\}$\;
$\widehat{\mathcal A}_\delta^{(h)}\leftarrow\{\pi:\widehat R_H(\pi)-\delta_R\ge r_{\min}^{(h)},\ \widehat C_H(\pi)+\delta_C\le c_{\max}^{(h)}\}$\;
Classify remaining candidates as certified unacceptable or unresolved using Proposition~\ref{prop:acceptability-margins}\;
\If{$\widehat{\mathcal A}_\delta^{(h)}=\varnothing$}{
  \Return{no IDS certified under the declared information and constraints}\;
}
$\widehat{\mathcal N}^{\rm model}_{\epsilon,\delta}\leftarrow
\{\pi\in\widehat{\mathcal A}_\delta^{(h)}:
\widehat J(\pi)\ge\sup_{\pi'\in\widehat{\mathcal A}_\delta^{(h)}}\widehat J(\pi')-\epsilon\}$\;
\If{$\Pi$ is finite and simultaneous objective intervals are available}{
  $\mathcal P\leftarrow\{\pi:\pi\text{ is not certified unacceptable}\}$\;
  $U^\star\leftarrow\max_{\pi\in\mathcal P}U_J(\pi)$\;
  $\widehat{\mathcal N}^{\rm cert}_{\epsilon,\delta}\leftarrow
  \{\pi\in\widehat{\mathcal A}_\delta^{(h)}:U^\star-L_J(\pi)\le\epsilon\}$\;
}
\Return{operators, diagnostics, classification regions, both near-optimal sets when defined, and a certified reference}\;
\end{algorithm2e*}

The propagation step may use marginal $q_t,v_t$ for pointwise criteria or an augmented joint model for hitting and cumulative events. Priors can regularize the identification stage, but they do not remove observational equivalence created by insufficient excitation. The model-relative near-optimal calculation is performed only after structural feasibility and conservative acceptability have been imposed. The stronger finite-class certificate additionally uses Proposition~\ref{prop:finite-epsilon-certificate}; it may be empty even when certified acceptable references exist.

Algorithm~1 separates five conclusions that should not be conflated. A successful stochastic fit establishes a predictive representation only on the excited and validated domain. Membership in $\widehat{\mathcal A}_0^{(h)}$ establishes nominal acceptability under the fitted model, whereas a nonempty $\widehat{\mathcal A}_\delta^{(h)}$ establishes that at least one candidate is certified under the declared information, posture, policy class, horizon, and margins. Candidates between the acceptance and rejection margins remain unresolved. The set $\widehat{\mathcal N}^{\rm model}_{\epsilon,\delta}$ ranks certified candidates using the fitted criterion, while $\widehat{\mathcal N}^{\rm cert}_{\epsilon,\delta}$ supplies the stronger finite-class guarantee when objective intervals permit it. Neither repairs an inadequate representation or converts an unacceptable policy into an acceptable one. This ordering makes the returned diagnostics part of the decision result rather than ancillary numerical output.

The two early returns also have different meanings. Failure in Stage I leaves the decision unresolved because the available representation cannot support classification near the relevant boundary. Failure in Stage II is stronger but still conditional: the tested policy class contains no certified candidate under the present constraints. In practice, the former calls for new excitation, observations, regime labels, graph revision, or state enrichment. The latter calls for revising the implementable policy class, resource limits, or negotiated thresholds only when such a revision is substantively justified. Keeping both outcomes explicit prevents an optimizer from silently compensating for an identification failure.

\section{Scope of switching and enrichment}
The regime index and the enrichment index serve different purposes. Switching selects among locally affine stochastic maps as operating conditions change; in the fully observed case, pre-action state-dependent switching produces the piecewise-affine nonlinear evolution described above. Enrichment refines a state when residence history or additional observations reveal decision-relevant internal levels. It should be retained only when the refinement changes prediction near an acceptance boundary or permits a less costly acceptable strategy. Spectral and pseudospectral diagnostics can support that assessment, but they do not by themselves establish predictive closure or decision sufficiency~\cite{vides2026localization}.

Operationally, the current regime may be prescribed, inferred before control selection, or generated by an observed pre-action state--context rule. An action-dependent rule instead updates the next regime, as specified by the causal convention in Eq.~\eqref{eq:causal-order}. These cases produce the same conditional propagation once $\nu_t$ is fixed, but they carry different uncertainty. If regime classification is estimated, its error must enter validation or the acceptance margins; it cannot be absorbed into the local operator error without justification. Similarly, enrichment should be evaluated under the same downstream criterion used for the decision. A refinement that changes eigenvalue or pseudospectral diagnostics but leaves every relevant return and cost classification unchanged has diagnostic value, yet supplies no decision advantage in the declared problem.

\section{Reproducible numerical laboratories}
All parameters below are synthetic. No physiological or commercial calibration is asserted. The known-matrix calibration and the three application prototypes use distinct parameterizations and outcome definitions; their absolute scores are therefore not directly comparable. Replication bars and boxes report descriptive medians and quartiles unless stated otherwise; they are not guarantees about model validity.

\subsection{Known switched matrices and SSRC calibration}
The new primary calibration instantiates Eq.~\eqref{eq:switched} with four internal states, three inputs, three return categories, three cost categories, and two observed regimes. Mixing coefficients are $(0.65,0.70,0.80)$. Sparse graph supports and all true matrices are provided in the reproducibility package. For each regime, 240 independent probability-state/input pairs are generated by controlled initialization; input probabilities are exact. Output counts use ensembles of 2000. Thus this experiment tests output sampling noise, not errors in measured inputs or unknown regime classification.

Using noiseless outputs, the maximum induced-$\ell^1$ matrix error is about $10^{-14}$ for the declared sparse supports. With multinomial outputs, held-out mean probability errors in the reported fit range from approximately 0.0021 to 0.0038. Twenty independent noisy fits probe variability under this same design, as shown in Fig.~\ref{fig:ssrc}. These observations support numerical adequacy for this example; the normal-equation construction and its conditioning do not justify a universal stability claim.

For control, fix a 12-step schedule alternating regimes in blocks of three and optimize constant allocations subject to $u_1+u_2\le0.50$, $0.10\le u_2\le0.15$. The target criterion is the last contemporaneous return probability; risk is the maximum of the per-step undesirable-cost marginal, not a probability of ever encountering that cost. Thresholds are 0.25 and 0.35. The response is affine in constant $u$, allowing exact polygon clipping and vertex optimization rather than a sampled allocation grid.

Known and identified models select $(0.35,0.15,0.50)$. The identified model predicts target probability 0.35660; evaluation under the known model gives 0.35699 and maximum pointwise cost probability 0.16606. With $\epsilon=0.02$, the entire feasible allocation polygon is nominally near-optimal in both models. Because the generator is known, this is an ex post model comparison, not a statistical certificate over the continuous polygon; obtaining simultaneous data-derived margins for that class is left open. This broad negotiation region is a result of this example's objective variation and thresholds, not an assumed neighborhood around an isolated optimum.

\subsection{Reproducible personal and inventory prototypes}
The first compact prototype uses four preparation states (low, developing, ready, and fatigued), three resource allocations (recovery, sustainable training, and demanding training), and ordinary or assessment regimes. Synthetic collective frequencies emulate summaries that could be constructed from personal records, wearable-derived categories, or a suitably modeled community prior. Such pooled evidence is informative only under an explicit transportability assumption; the example does not treat people as automatically exchangeable.

The second prototype uses four inventory states (healthy, normal, low, and pending/delayed), three allocations (flexible supplier, reliable supplier, and expedite/buffer), and ordinary or stressed demand~\cite{zipkin2000inventory}. Supplier-delay information is represented at this coarse resolution; delivery-age enrichment remains appropriate when elapsed time changes the arrival hazard enough to alter the negotiated choice.

Both laboratories generate 260 state--input pairs per regime and output frequencies from ensembles of 2500 cases. Graph-supported nonnegative least squares imposes unit column sums as exact constraints. The constrained diagnostics apply Eq.~\eqref{eq:restricted-design}: personal-state maps have zero restricted nullity, while their return and cost maps have nullities one and two; inventory-state maps likewise have zero restricted nullity, while their return and cost maps have nullity two in each regime. Nevertheless, held-out mean $\ell^1$ prediction errors range from 0.0027 to 0.0060. This distinction matters: accurate predictions on the excited simplex do not establish unique recovery of every admissible edge weight. Direct worst-case bounds formed from full induced matrix errors are consequently too conservative here and classify every selected policy as unresolved.

To expose the certification logic without disguising that limitation, the synthetic generator is used only for an oracle diagnostic on each declared finite candidate class. The simultaneous discrepancies
\begin{equation}
\delta_R=\max_{\pi\in\Pi_f}|R_H(\pi)-\widehat R_H(\pi)|,\qquad
\delta_C=\max_{\pi\in\Pi_f}|C_H(\pi)-\widehat C_H(\pi)|
\end{equation}
are exact over the enumerated class $\Pi_f$. They certify the reported finite-class classifications deterministically for these synthetic experiments, but are not estimable guarantees for an applied data set.

Personal allocations satisfy sustainable-training shares in $[0.25,0.55]$ and demanding-training shares no larger than 0.25. Acceptance requires terminal high-success probability at least 0.25 and maximum pointwise unsustainable-burden probability at most 0.18. The finite-class oracle discrepancies are $\delta_R=0.001240$ and $\delta_C=0.001116$. Of 195 grid-feasible allocations, 44 are nominally acceptable, 39 are certified acceptable, 10 are unresolved, and 146 are certified unacceptable. Among the certified policies, 23 lie within $\epsilon=0.02$ of the best fitted return; 18 also satisfy the stronger finite-class certificate in Proposition~\ref{prop:finite-epsilon-certificate}. The selected certified reference is $(0.28,0.54,0.18)$ for recovery, sustainable training, and demanding training. Its fitted terminal target probability is 0.27783 and its conservative lower bound is 0.27659; its fitted maximum burden is 0.17703 and its conservative upper bound is 0.17815. Its limiting constraint is cost, giving $\rho=0.00185$.

Inventory allocations require the flexible and reliable shares to total at most 0.60, with reliable participation in $[0.15,0.30]$. Both declared postures require terminal high-service probability at least 0.48 and use $\epsilon=0.015$; the prudential posture sets the maximum pointwise unsustainable-cost probability at 0.22, whereas the continuity-operational posture uses 0.24. The posture is selected before optimization from the business context; the second ceiling is not introduced after a failed certificate. Here $\delta_R=0.003102$ and $\delta_C=0.001327$.

Under the prudential posture, of 616 feasible allocations, 50 are nominally acceptable, 28 are certified acceptable, 53 are unresolved, and 535 are certified unacceptable. Twenty-five certified policies are model-relative $\epsilon$-near-optimal, while one satisfies the stronger finite-class $\epsilon$-optimality certificate. The nominal maximizer $(0.14,0.30,0.56)$ is unresolved because its conservative cost upper bound is 0.2203. The selected certified reference is $(0.15,0.30,0.55)$, with fitted terminal high-service probability 0.49897, conservative lower bound 0.49586, fitted maximum cost 0.21778, and conservative upper bound 0.21910. Its cost-limited certified slack is $\rho=0.00090$.

Under the continuity-operational posture, 303 allocations are nominally acceptable, 275 are certified acceptable, 57 are unresolved, and 284 are certified unacceptable. Thirty-one certified policies are model-relative $\epsilon$-near-optimal and ten satisfy the stronger certificate. The selected operational reference is $(0,0.28,0.72)$, with fitted terminal high-service probability 0.55129, conservative lower bound 0.54819, fitted maximum cost 0.23832, and conservative upper bound 0.23965. Its certified slack is $\rho=0.00035$. Thus a predeclared willingness to tolerate two additional probability points of pointwise cost risk permits a materially more active high-service allocation while retaining a positive certificate. These thresholds are illustrative preferences and the cost constraints concern marginal probabilities, not pathwise hitting events.

Figure~\ref{fig:didactic} is the central numerical synthesis. Panel (a) distinguishes acceptability certification, model-relative near-optimality, and the stronger finite-class certificate in the personal laboratory. Panel (b) connects the selected personal allocation to its time-dependent return and burden localizations. Panel (c) contrasts the prudential and continuity-operational inventory postures, including their different certified references and stronger $\epsilon$-optimal subsets. Panel (d) maps both references back to implementable supplier shares. Together, the panels display the proposed IDS chain from stochastic representation to a context-declared negotiable decision under residual uncertainty.

\begin{figure*}[t]
\centering
\includegraphics[width=\textwidth]{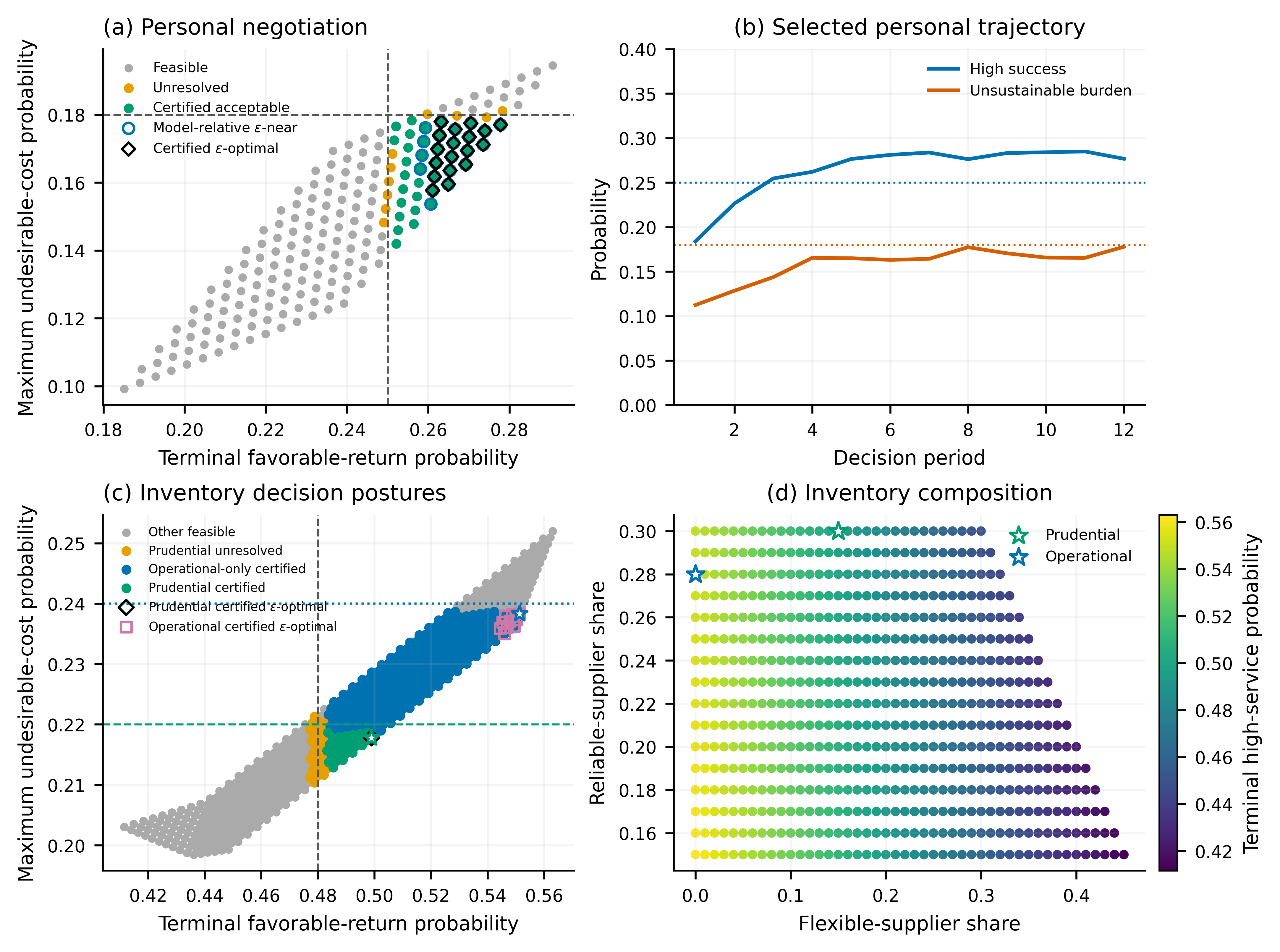}
\caption{Reproducible compact IDS prototypes with finite-class oracle bounds. (a) Personal return--burden classification; open blue circles show model-relative near-optimal policies and open black diamonds show the stronger certified $\epsilon$-optimal subset. (b) Trajectories under the selected certified personal allocation; dotted lines show its thresholds. (c) Inventory classification under the prudential ceiling $0.22$ and continuity-operational ceiling $0.24$: green points are prudentially certified, blue points are certified only under the operational posture, diamonds and squares mark the respective strongly certified $\epsilon$-optimal subsets, and stars mark their references. (d) Both inventory references in allocation space. The posture is declared before policy comparison. Bounds are exact synthetic finite-class diagnostics, not empirical coverage guarantees.}
\label{fig:didactic}
\end{figure*}

\begin{figure*}[t]
\centering
\includegraphics[width=\textwidth]{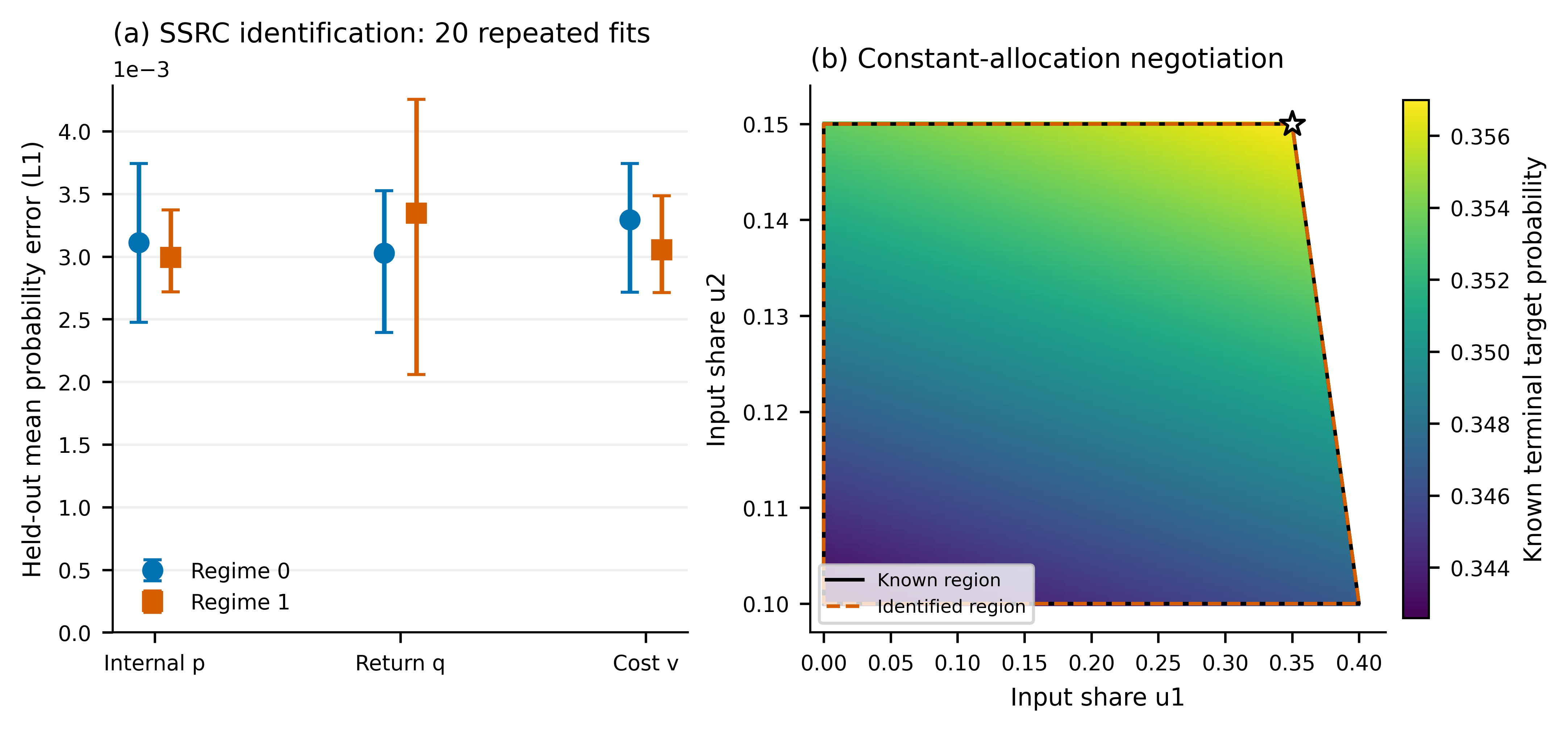}
\caption{Primary switched-additive calibration. (a) Median and interquartile range of held-out probability errors over 20 noisy SSRC fits, shown by regime and output. These bars are variability summaries, not confidence guarantees. (b) Constant-allocation negotiation polygon colored by known-model target probability. Known and identified feasible boundaries coincide, and the star marks their common selected optimum. The full polygon is $0.02$-near-optimal under both models.}
\label{fig:ssrc}
\end{figure*}

\begin{figure*}[t]
\centering
\includegraphics[width=\textwidth]{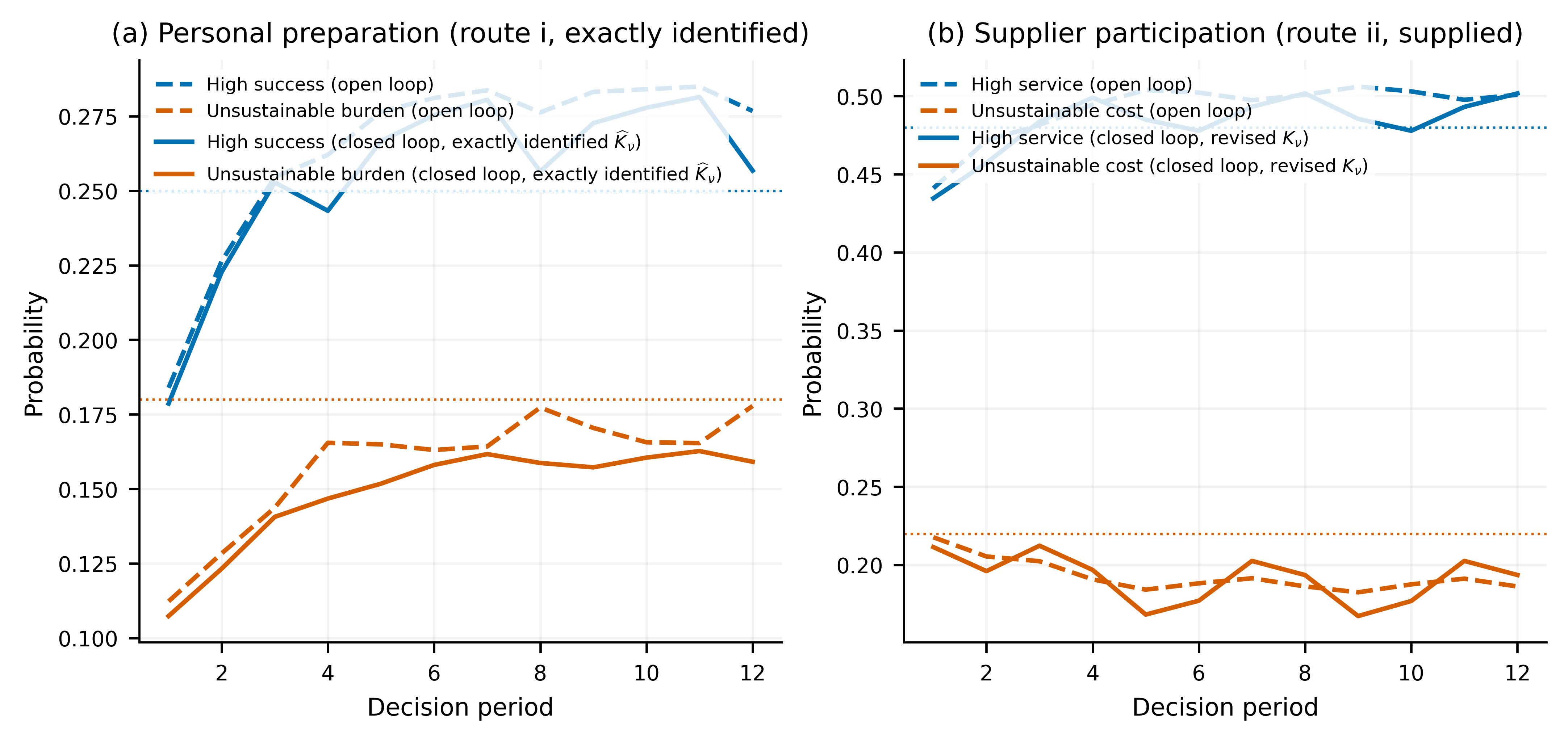}
\caption{Certified open-loop references versus certified closed-loop feedback. (a) Personal preparation: constant allocation $u^\star$ (dashed) and the exactly identified structured controller $\widehat K_\nu$ (route i, solid). (b) Supplier participation under the prudential posture: constant allocation $u^\star$ (dashed) and the revised operator-supplied rule (route ii, solid). Dotted lines mark the terminal favorable-return thresholds and applicable pointwise cost ceilings. Certification uses exact oracle discrepancies over each declared finite synthetic class; these are diagnostic margins, not empirical coverage guarantees.}
\label{fig:closedloop}
\end{figure*}

\subsection{Stochastic feedback closure: two certified feedback laboratories}
Sec.~II.D and the compact prototypes above are combined here: for each prototype, a regime-indexed stochastic quadratic feedback controller $u_t=K_{\nu_t}\eth_2(\bar p_t)$ is constructed and evaluated against the same acceptance thresholds used for the personal and prudential-inventory constant-allocation references. In these fully observed synthetic laboratories, $\bar p_t=p_t$. Route (i) (estimation from data) is used for the personal-preparation prototype and route (ii) (supplied candidates) for the inventory prototype. Exact simultaneous oracle discrepancies are computed over the two declared feedback candidates in each laboratory.

For the personal-preparation prototype, the controller is parameterized by a stochastic per-state preference matrix $A_\nu\in\mathbb R^{3\times4}$. Its canonical degree-two realization satisfies $K_\nu^\star\eth_2(p)=A_\nu p$ on $\Delta^3$. Exact synthetic pairs include the simplex vertices and interior samples, giving a rank-$4$ design for the minimal parametrization. The samplewise Frobenius least-squares fit in Eq.~\eqref{eq:minimal-feedback-identification} recovers $A_\nu$ with maximum coefficient error $3.33\times10^{-16}$ and reconstructs $\widehat K_\nu$ with worst-case held-out functional error $5.59\times10^{-16}$. The generating and identified closed-loop trajectories differ by at most $2.22\times10^{-16}$, so a numerical tolerance of $10^{-12}$ is used when comparing the exact benchmark. Thus the structured feedback law, rather than every coefficient of an unrestricted dense $3\times14$ representation, is exactly identifiable on the declared domain. Over the declared two-candidate feedback family, $\delta_R=0.000883$ and $\delta_C=0.001135$. For the identified controller, $\widehat R_H=0.25764$ and $\widehat C_H=0.16384$; the conservative values are 0.25676 and 0.16497, respectively. Hence $\rho=0.00676$, and this finite family is $\eta$-certifiably feasible for every $0\le\eta\le0.00676$.

For the inventory prototype, the declared two-candidate feedback family has oracle discrepancies $\delta_R=0.002495$ and $\delta_C=0.000855$. A first controller supplied from a stated purchasing-team escalation rule (route ii, no estimation) is certified unacceptable: its fitted terminal service probability is 0.42721 and even its upper bound, 0.42970, is below 0.48; its fitted maximum cost is 0.19360, and $\rho=-0.05529$. A single revision of the same rule, shifting escalation weight toward the reliable and expedite suppliers, is certified acceptable: $\widehat R_H=0.49997$ with lower bound 0.49748, and $\widehat C_H=0.21314$ with upper bound 0.21400. Evaluation under the generator gives 0.50172 and 0.21229, respectively. The revised rule has $\rho=0.00600$, so the tested family is $\eta$-certifiably feasible for $0\le\eta\le0.00600$. This is the intended use of certification for a supplied candidate: identical to an estimated one, and no more exempt from failing it.

Figure~\ref{fig:closedloop} compares, for each prototype, the certified closed-loop trajectory against the certified open-loop constant-allocation reference of Fig.~\ref{fig:didactic}(b,d). In both cases the pointwise cost ceiling is satisfied throughout the horizon and the favorable-return threshold is satisfied at the terminal period. The personal-preparation controller is the exactly identified realization described above and maintains a visibly larger margin below the burden ceiling than the open-loop reference. The revised inventory rule closely tracks its open-loop reference after certification. Neither comparison establishes that closed-loop feedback dominates constant allocation in general; it shows that both policy forms contain certified alternatives in these laboratories.

\subsection{Noisy feedback identification for customer-retention allocation}
The exact personal benchmark is complemented by a third synthetic laboratory in which feedback observations are noisy. This example concerns sequential allocation of a limited retention-intervention capacity, not standalone churn prediction. Four portfolio localizations represent stable, vulnerable, at-risk, and disengaged customers; three allocation categories represent routine service, personalized contact, and retention incentives. Return categories distinguish lost, retained, and expanded relationships, while cost categories distinguish low, moderate, and excessive intervention intensity. Ordinary and adverse commercial conditions define two observed regimes.

For a prescribed mean allocation $u_t=A_{\nu_t}p_t$, the observed executed share is generated by
\begin{equation}
N\widetilde u_t\sim\operatorname{Multinomial}(N,u_t).
\label{eq:noisy-executed-allocation}
\end{equation}
Return and cost localizations are observed analogously as multinomial event frequencies conditional on $(p_t,\widetilde u_t)$. The state localization and regime are observed without noise, deliberately excluding partial observation from this experiment. For each regime, $A_\nu$ is estimated from 160 state--allocation pairs by the constrained samplewise loss in Eq.~\eqref{eq:minimal-feedback-identification}; graph-supported return and cost maps are estimated from the corresponding event frequencies. Thirty independent replications are performed for $N\in\{100,500,2500\}$, with 2000 independently sampled state localizations per regime for functional evaluation.

The median held-out mean $\ell^1$ error of the feedback law decreases from 0.01606 to 0.00714 and 0.00339 as $N$ increases. Median on-trajectory oracle discrepancies in terminal return decrease from 0.00285 to 0.00128 and 0.00053; the corresponding maximum-cost discrepancies decrease from 0.00142 to 0.00084 and 0.00043. With illustrative thresholds 0.183 for terminal favorable return and 0.112 for maximum pointwise excessive cost, 17, 18, and 27 of the 30 fitted policies are oracle-certified acceptable, while 13, 12, and 3 remain unresolved. None is certified unacceptable.

Figure~\ref{fig:noisy-retention} separates policy-function recovery from decision classification. The minimal controller design has zero restricted nullity, whereas each graph-supported return and cost design has restricted nullity two in both regimes. The return and cost blocks fitted from closed-loop observations consequently retain substantially larger error on independently sampled off-policy actions, approximately 0.04--0.07 in mean $\ell^1$ error across these designs. This does not contradict the decreasing controller error. Two limitations coexist: the additive $W$--$V$ representation contains feasible observationally equivalent directions, and actions generated from state do not independently cover arbitrary state--action combinations. Thus larger event counts improve the identified feedback law and its on-policy decision resolution without establishing unique coefficient recovery or unrestricted off-policy prediction. The exact discrepancies used for classification remain synthetic oracle diagnostics, not empirical guarantees for a commercial application.

\begin{figure*}[t]
\centering
\includegraphics[width=\textwidth]{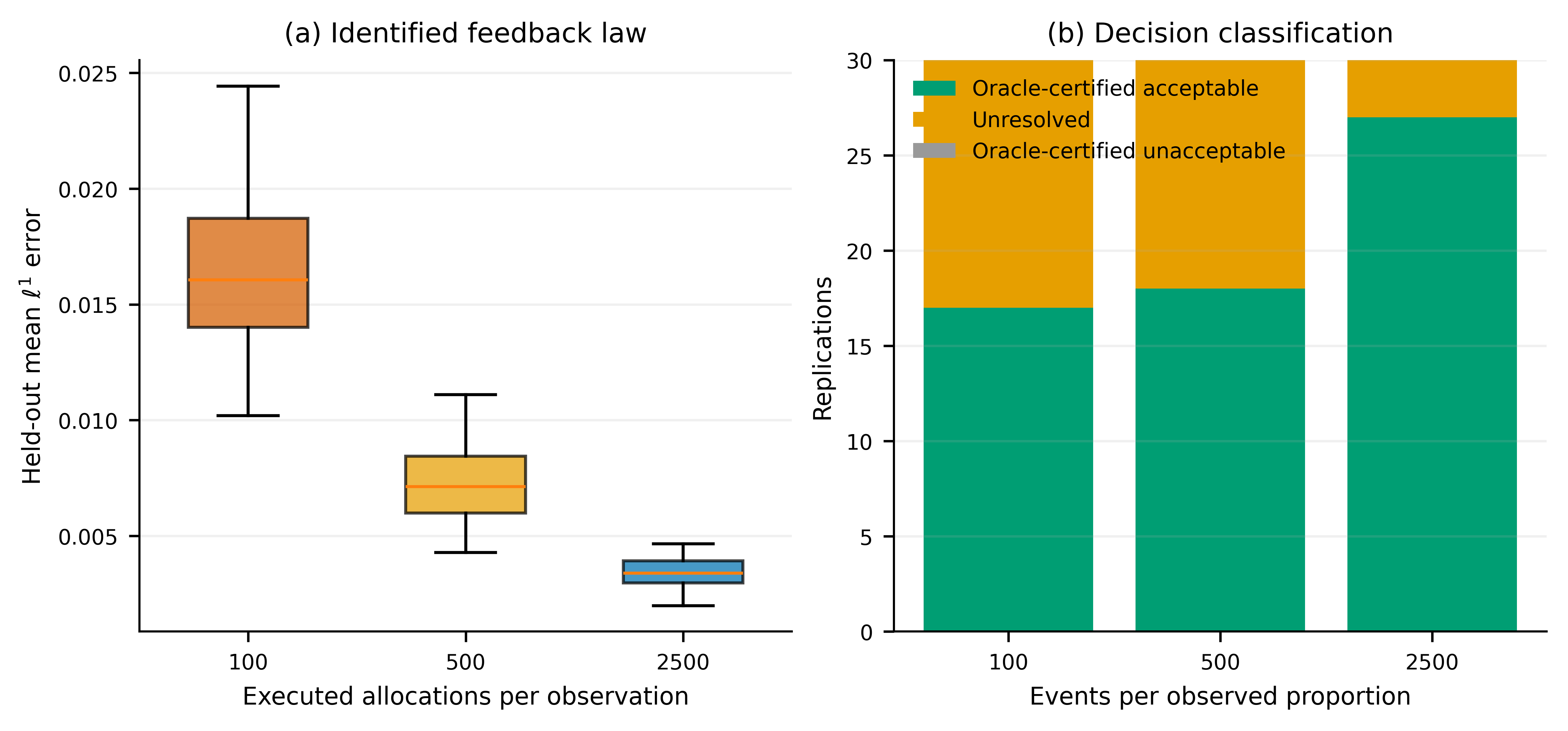}
\caption{Noisy customer-retention feedback laboratory over 30 independent replications. (a) Held-out functional error of the identified stochastic allocation law for multinomial sample sizes $N=100,500,2500$. Boxes show quartiles and medians; whiskers use the standard 1.5 interquartile-range rule and outliers are omitted. (b) Classification of the fitted policies under exact two-candidate oracle discrepancies in each synthetic replication. Increasing event counts reduce unresolved classifications, but these oracle margins are not data-derived commercial guarantees.}
\label{fig:noisy-retention}
\end{figure*}

\section{Discussion}
The central claim is not that every dynamic strategy possesses a unique stochastic matrix representation. It is that a structured probability representation can be sufficient for a declared decision even when some coefficients remain observationally equivalent. That sufficiency is conditional on the operating regimes, the excited domain, the policy class, the horizon, the events, and the acceptance margins. Changing any of these elements may turn a previously harmless null direction into a decision-relevant discrepancy. Identification diagnostics are therefore part of the decision result rather than a preliminary numerical formality.

This perspective separates three notions often conflated in data-driven control. Structural reduction removes repeated tensor words before fitting. Identifiability asks whether the observations distinguish the remaining graph-admissible stochastic directions. Decision sufficiency asks whether any unresolved directions can change the classification of relevant policies. The personal benchmark illustrates the distinction cleanly: its $14$-coordinate controller realization is already SSRC-reduced, yet cross-degree relations induced by the simplex make that ambient realization nonunique. The minimal stochastic matrix $A_\nu$ remains identifiable and determines the same allocation function. Thus a nonunique realization need not imply an ambiguous implemented policy.

Probability readouts also place this work near stochastic reservoir computing, where controlled Markov evolution can generate rich probabilistic features and finite sampling creates a hardware--precision tradeoff~\cite{ehlers2025stochastic}. The emphasis here is complementary. Universality concerns the expressive capacity of an architecture over a function class; IDS decision sufficiency concerns preservation of a particular acceptability classification under structural and statistical uncertainty. A universally expressive class need not be identifiable from a given closed-loop experiment, while a nonunique fitted representation may still be sufficient for a narrow decision. The retention laboratory demonstrates the latter possibility: increasing multinomial sample size improves the feedback function and reduces unresolved classifications, but does not remove intrinsic $W$--$V$ equivalence or supply independent coverage of arbitrary actions.

The same distinction matters for collective information in human-facing systems. Interindividual and intraindividual variability may represent adaptive or persistent organization rather than measurement noise~\cite{forkel2026neurovariability}. Conversely, finite observations, devices, protocols, and classification rules introduce genuine measurement uncertainty. Pooling these components into one residual would obscure both transportability and certification. The personal laboratory therefore represents probability summaries for a declared population or individual information source; it does not assume exchangeability, infer neurophysiological mechanisms, or convert population frequencies automatically into individual guarantees. The notion of functional vicariance discussed in neurovariability provides a useful analogy: different internal configurations may support a similar function. In IDS terms, structurally different policies or operators can be outcome-equivalent on selected times and coordinates. This is a conceptual connection, not an asserted biological correspondence.

Certification then determines what may responsibly be concluded from the representation. A fitted-model optimum is only nominal. An acceptable certificate additionally requires justified return and cost margins, while certified $\epsilon$-optimality requires simultaneous objective intervals over the declared finite class. The two inventory postures show why thresholds must be fixed from auditable pre-action conditions rather than chosen after inspecting the optimizer. The feedback experiments make the same point about policy provenance: an estimated controller and an operator-supplied rule face identical tests. The first supplier rule is rejected before a revised rule is certified, whereas the noisy retention policies may remain unresolved instead of being forced into an accept/reject label.

These distinctions also delimit the relation to stochastic optimal control. Chance-constrained and stochastic model predictive control methods optimize performance subject to probabilistic constraints~\cite{mesbah2016smpc}. The present framework does not replace those methods or establish recursive feasibility and closed-loop stability in their general sense. It supplies a structured representation, a conservative classification interface, and a negotiable set from which optimization or operational judgment may select. Likewise, regime switching supplies local nonlinear expressiveness but does not by itself validate a regime partition. State enrichment, additional regimes, or more complex feedback are warranted only when they improve prediction or resolve a decision boundary.

The evidence remains computational and synthetic. Exact oracle discrepancies make the finite-class logic auditable, but an application requires data-derived simultaneous bounds, observation and filtering models, validation of regime labels, and explicit separation of sampling error from structural heterogeneity. Closed-loop deployment also requires stability and safety analysis beyond the stochastic closure proved here. Joint trajectory models are needed for hitting, persistence, and cumulative-cost events; marginal output distributions are insufficient. Natural extensions include active excitation under operational constraints, posterior or set-valued propagation over observational-equivalence classes, outcome-level controller equivalence, minimal-complexity selection among certified policies, information-acquisition costs, and adaptive renegotiation.

\section*{Data Availability Statement}
The numerical examples in this work are synthetic and require no external empirical data. The programs, executable IDS prototype notebooks, saved numerical results, and figure-generation materials supporting this study will be made available in due course through the IntegratedDynamicStrategies GitHub repository~\cite{vides2026integrateddynamicstrategies}. The extended manuscript and additional didactic experiments will also be maintained in that repository. All numerical examples reported in this work are synthetic and do not incorporate proprietary, client, or personal data.

\par\medskip

\section*{Acknowledgments}
The author thanks colleagues and practitioners for observations that helped
motivate the inventory and customer-retention laboratories, and gratefully
acknowledges the institutional support provided by Universidad Nacional
Autónoma de Honduras (UNAH).

\par\medskip

\balance
\bibliographystyle{unsrtnat}
\bibliography{references}

@misc{vides2026localization,
  author        = {Vides, Fredy},
  title         = {Identifying Probability Localization Dynamics via Structured Stochastic Liftings},
  year          = {2026},
  eprint        = {2608.22686},
  archivePrefix = {arXiv},
  primaryClass  = {math.DS},
  doi           = {10.48550/arXiv.2608.22686},
  url           = {https://arxiv.org/abs/2608.22686v1},
  note          = {Version 1}
}

@article{jia2016,
  author  = {Jia, Chen and Jiang, Da-Quan and Qian, Min-Ping},
  title   = {Cycle Symmetries and Circulation Fluctuations for Discrete-Time and Continuous-Time {Markov} Chains},
  journal = {The Annals of Applied Probability},
  volume  = {26},
  number  = {4},
  pages   = {2454--2493},
  year    = {2016},
  doi     = {10.1214/15-AAP1152}
}

@article{banegas2025ssrc,
  author  = {Banegas, Lendy and Vides, Fredy},
  title   = {Stochastically Structured Reservoir Computers for Financial and Economic System Identification},
  journal = {IFAC-PapersOnLine},
  volume  = {59},
  number  = {36},
  pages   = {100--105},
  year    = {2025},
  doi     = {10.1016/j.ifacol.2026.03.018},
  url     = {https://www.sciencedirect.com/science/article/pii/S2405896326001084}
}

@article{ehlers2025stochastic,
  author  = {Ehlers, Peter J. and Nurdin, Hendra I. and Soh, Daniel},
  title   = {Stochastic Reservoir Computers},
  journal = {Nature Communications},
  volume  = {16},
  pages   = {3070},
  year    = {2025},
  doi     = {10.1038/s41467-025-58349-6},
  url     = {https://doi.org/10.1038/s41467-025-58349-6}
}

@article{forkel2026neurovariability,
  author  = {Forkel, Stephanie J. and Dulyan, Lilit and Schilling, Kurt G. and Thiebaut de Schotten, Michel},
  title   = {Neurovariability as a Signature of Adaptive Brain Function},
  journal = {Brain Structure and Function},
  volume  = {231},
  pages   = {83},
  year    = {2026},
  doi     = {10.1007/s00429-026-03138-0},
  url     = {https://doi.org/10.1007/s00429-026-03138-0}
}

@book{ljung1999system,
  author    = {Ljung, Lennart},
  title     = {System Identification: Theory for the User},
  edition   = {2},
  publisher = {Prentice Hall},
  address   = {Upper Saddle River, NJ},
  year      = {1999}
}

@article{forssell1999closedloop,
  author  = {Forssell, Urban and Ljung, Lennart},
  title   = {Closed-Loop Identification Revisited},
  journal = {Automatica},
  volume  = {35},
  number  = {7},
  pages   = {1215--1241},
  year    = {1999},
  doi     = {10.1016/S0005-1098(99)00022-9}
}

@article{paoletti2007hybrid,
  author  = {Paoletti, Simone and Juloski, Aleksandar Lj. and Ferrari-Trecate, Giancarlo and Vidal, Ren{\'e}},
  title   = {Identification of Hybrid Systems: A Tutorial},
  journal = {European Journal of Control},
  volume  = {13},
  number  = {2--3},
  pages   = {242--260},
  year    = {2007},
  doi     = {10.3166/ejc.13.242-260}
}

@article{mesbah2016smpc,
  author  = {Mesbah, Ali},
  title   = {Stochastic Model Predictive Control: An Overview and Perspectives for Future Research},
  journal = {IEEE Control Systems Magazine},
  volume  = {36},
  number  = {6},
  pages   = {30--44},
  year    = {2016},
  doi     = {10.1109/MCS.2016.2602087}
}

@misc{vides2026integrateddynamicstrategies,
  author       = {Vides, Fredy},
  title        = {{IntegratedDynamicStrategies}: Structured Stochastic Representations of Integrated Dynamic Strategies},
  year         = {2026},
  howpublished = {GitHub repository},
  note         = {\url{https://github.com/FredyVides/IntegratedDynamicStrategies}}
}

@article{simon1955satisficing,
  author  = {Simon, Herbert A.},
  title   = {A Behavioral Model of Rational Choice},
  journal = {Quarterly Journal of Economics},
  volume  = {69},
  number  = {1},
  pages   = {99--118},
  year    = {1955},
  doi     = {10.2307/1884852}
}

@article{nash1950bargaining,
  author  = {Nash, John F.},
  title   = {The Bargaining Problem},
  journal = {Econometrica},
  volume  = {18},
  number  = {2},
  pages   = {155--162},
  year    = {1950},
  doi     = {10.2307/1907266}
}

@article{hamilton1989regime,
  author  = {Hamilton, James D.},
  title   = {A New Approach to the Economic Analysis of Nonstationary Time Series and the Business Cycle},
  journal = {Econometrica},
  volume  = {57},
  number  = {2},
  pages   = {357--384},
  year    = {1989},
  doi     = {10.2307/1912559}
}

@book{liberzon2003switching,
  author    = {Liberzon, Daniel},
  title     = {Switching in Systems and Control},
  series    = {Systems \& Control: Foundations \& Applications},
  publisher = {Birkh\"auser},
  address   = {Boston},
  year      = {2003}
}

@book{lawson1974solving,
  author    = {Lawson, Charles L. and Hanson, Richard J.},
  title     = {Solving Least Squares Problems},
  publisher = {Prentice-Hall},
  address   = {Englewood Cliffs, NJ},
  year      = {1974}
}

@book{zipkin2000inventory,
  author    = {Zipkin, Paul H.},
  title     = {Foundations of Inventory Management},
  publisher = {McGraw-Hill},
  address   = {Boston},
  year      = {2000}
}

@book{miettinen1999nonlinear,
  author    = {Miettinen, Kaisa},
  title     = {Nonlinear Multiobjective Optimization},
  publisher = {Kluwer Academic Publishers},
  address   = {Boston},
  year      = {1999}
}
\end{document}